\documentclass[11pt]{article}
\usepackage[utf8]{inputenc}
\usepackage[T1]{fontenc}
\usepackage{authblk}
\usepackage{amsmath,amssymb,amsfonts,amsthm}
\usepackage{fullpage}
\usepackage{xspace}

\newtheorem{proposition}{Proposition}

\theoremstyle{definition}

\theoremstyle{remark}

\newcommand{\ket}[1]{\ensuremath{|#1\rangle}\xspace}
\newcommand{\bra}[1]{\ensuremath{\langle #1|}\xspace}
\newcommand{\Braket}[3]{\ensuremath{\bra{#1}#2\ket{#3}}\xspace}
\newcommand{\braket}[2]{\ensuremath{\langle #1|#2\rangle}\xspace}
\numberwithin{equation}{section}
\title{$d$-Orthogonal polynomials and $\mathfrak{su}(2)$}
\date{}
\author[1]{Vincent X. Genest\thanks{genestvi@crm.umontreal.ca}}
\author[1]{Luc Vinet\thanks{luc.vinet@umontreal.ca}}
\author[2]{Alexei Zhedanov\thanks{zhedanov@kinetic.ac.donetsk.ua}}
\affil[1]{Centre de recherches math\'ematiques, Universit\'e de Montr\'eal, C.P. 6128, Succursale Centre-ville, Montr\'eal, Qu\'ebec, H3C 3J7, Canada}
\affil[2]{Donetsk Institute for Physics and Technology, Donetsk 83114, Ukraine}
\begin{document}
\maketitle
\thispagestyle{empty}
\hrule
\begin{abstract}
Two families of $d$-orthogonal polynomials related to $\mathfrak{su}(2)$ are identified and studied. The algebraic setting allows for their full characterization (explicit expressions, recurrence relations, difference equations, generating functions, etc.). In the limit where $\mathfrak{su}(2)$ contracts to the Heisenberg-Weyl algebra $\mathfrak{h}_1$, the polynomials tend to the standard Meixner polynomials and $d$-Charlier polynomials, respectively.
\end{abstract}
\bigskip 

\hrule
\section{Introduction}
We identify in this paper two families of $d$-orthogonal polynomials that are associated to $\mathfrak{su}(2)$. When available, algebraic models for orthogonal polynomials provide a cogent framework for the characterization of these special functions. They also point to the likelihood of seeing the corresponding polynomials occur in the description of physical systems whose symmetry generators form the algebra in question.

 $d$-orthogonal polynomials generalize the standard orthogonal polynomials in that they obey higher recurrence relations. They will be defined below and have been seen to possess various applications \cite{deBruin-1985,deBruin-1990,Iseghem-1987}. Recently, two of us have uncovered the connection between $d$-Charlier polynomials and the Heisenberg algebra. Here we pursue this exploration of $d$-orthogonal polynomials related to Lie algebras by considering the case of $\mathfrak{su}(2)$. Remarkably, two hypergeometric families of such polynomials will be identified and characterized.
\subsection{$d$-Orthogonal polynomials}
The monic $d$-orthogonal polynomials $\widehat{P}_{n}(k)$ of degree $n$ ($\widehat{P}_{n}(k)=k^{n}+\ldots$) can be defined by the recurrence relation \cite{Maroni-1989}
\begin{equation}
\label{recu-ortho}
\widehat{P}_{n+1}(k)=k\widehat{P}_{n}(k)-\sum_{\mu=0}^{d}a_{n,n-\mu}\widehat{P}_{n-\mu}(k),
\end{equation}
of order $d+1$ with complex coefficients $a_{n,m}$; the initial conditions are $\widehat{P}_{n}=0$ if $n<0$ and $\widehat{P}_{0}=1$. It is assumed that $a_{n,n-d}\neq 0$ (non-degeneracy condition).

When $d=1$, it is known that under the condition $c_{n}\neq 0$,  the polynomials  satisfying three-term recurrence relations
\begin{equation*}
\widehat{P}_{n+1}(k)=k\widehat{P}_{n}(k)-b_{n}\widehat{P}_{n}(k)-c_{n}\widehat{P}_{n-1}(k),
\end{equation*}
are orthogonal with respect to a linear functional $\sigma$ 
\begin{equation*}
\langle \sigma, \widehat{P}_{n}(k)\widehat{P}_{m}(k)\rangle=\delta_{nm},
\end{equation*}
defined on the space of all polynomials.

When $d>1$, the polynomials $\widehat{P}_{n}(k)$ obey \emph{vector orthogonality relations}. This means that there exists a set of $d$ linear functionals $\sigma_{i}$ for $i=0,\ldots,\,d-1$ such that the following relations hold:
\begin{align*}
\langle \sigma_{i},\widehat{P}_{n}\widehat{P}_m\rangle &=0,\;\text{if}\;m>dn+i, \nonumber\\
\langle \sigma_{i},\widehat{P}_{n}\widehat{P}_{dn+i}\rangle &\neq 0,\;\text{if}\;n\geqslant 0.
\end{align*}

\subsection{$d$-Orthogonal polynomials as generalized hypergeometric functions}
Of particular interest are $d$-orthogonal polynomials that can be expressed in terms of generalized hypergeometric functions (see for instance \cite{VanAssche-2011, VanAssche-2001,deBruin-1984}). These functions are denoted ${}_pF_{q}$ and are defined by
\begin{equation*}
{}_pF_{q}\left[\begin{aligned}\{a_p\}\\\{b_q\}\end{aligned};\frac{1}{c}\right]:=\sum_{\mu=0}^{\infty}\frac{(a_1)_{\mu}\cdots(a_p)_{\mu}}{(b_1)_{\mu}\cdots(b_q)_{\mu}}\frac{c^{-\mu}}{\mu!},
\end{equation*}
where $(m)_0=1$ and $(m)_{k}=(m)(m+1)\cdots(m+k-1)$ stands for the Pochhammer symbol. In the case where one of the $a_{i}$'s is a negative integer, say $a_1=-n$ for $n\in \mathbb{N}$, the series truncates at $\mu=n$ and we can write
\begin{align}
\label{polynomial-form}
{}_{1+s}F_{q}\left[ \begin{aligned}-n,&\;\{a_s\}\\ \{&b_q\}\end{aligned};\frac{1}{c}\right]=\sum_{\mu=0}^{n}\frac{(-n)_{\mu}(a_1)_{\mu}\cdots(a_s)_{\mu}}{(b_1)_{\mu}\cdots(b_q)_{\mu}}\frac{c^{-\mu}}{\mu!}.
\end{align}
If one of the $b_i$'s is also a negative integer, the corresponding sequence of polynomials is finite. 

The classification of $d$-orthogonal polynomials that have a hypergeometric representation of the form \eqref{polynomial-form} has been studied recently in \cite{Cheikh-2008}. The results are as follows.

Let $s\geqslant 1$ and let $\{a_s\}=\{a_1,\ldots,a_{s}\}$ be a set of $s$ polynomials of degree one in the variable $k$. This set is called $s$-separable if there is a polynomial $\pi(y)$ such that $\prod_{i=1}^{s}(a_i(k)+y)=[\prod_{i=1}^{s}a_i(k)]+\pi(y)$; an example of such $s$-separable set is $\{k\,e^{\frac{2\pi ix}{s}},\,x=0\ldots,\,s-1\}$. If the set $\{a_s\}$ is $s$-separable, then there exists only $2(d+1)$ classes of $d$-orthogonal polynomials of type \eqref{polynomial-form} corresponding to the cases:
\begin{enumerate}
\item $s=0,\ldots,d-1$ and $q=d$;
\item $s=d$ and $q=0,\ldots,d-1$;
\item $s=q=d$ and $c\neq1$;
\item $s=q=d+1$ and $c=1$ and $\sum_{i=0}^{d+1}a_{i}(0)-\sum_{i=1}^{d+1}b_i\notin\mathbb{N}$.
\end{enumerate}

Examples of $d$-orthogonal polynomials belonging to this classification have been found in \cite{Cheikh-2003,Vinet-2009}. We shall here provide more examples that are of particular interest as well as cases that fall outside the scope of the classification given above.
\subsection{Purpose and outline}
We investigate in this paper $d$-orthogonal polynomials associated to $\mathfrak{su}(2)$. We shall consider two operators $S$ and $Q$, each defined as the product of exponentials in the Lie algebra elements. We shall hence  determine the action of these operators on the canonical $(N+1)$-dimensional irreducible representation spaces of $\mathfrak{su}(2)$. In both cases, the corresponding matrix elements will be found to be expressible in terms of $d$-orthogonal polynomials, some of them belonging to the above-mentioned classification. The connection with  the Lie algebra $\mathfrak{su}(2)$ will be used to fully characterize the two families of $d$-orthogonal polynomials. The limit as $N\rightarrow \infty$, where $\mathfrak{su}(2)$ contracts to $\mathfrak{h}_1$, shall also be studied. In this limit, the polynomials are shown to tend on the one hand to the standard Meixner polynomials and on the other hand to $d$-Charlier polynomials.

The outline of the paper is as follows. In section $2$, we recall, for reference, basic facts about the $\mathfrak{su}(2)$ algebra and its representations. We also define a set of $\mathfrak{su}(2)$-coherent states and summarize how $\mathfrak{su}(2)$ contracts to the Heisenberg-Weyl algebra in the limit as $N\rightarrow \infty$. We then define the operators $S$ and $Q$ that shall be studied along with their matrix elements. The biorthogonality and recurrence relations of these matrix elements are made explicit from algebraic considerations. Results obtained in \cite{Vinet-2009} and \cite{Vinet-2011}  concerning the Meixner and $d$-Charlier polynomials shall also be recalled. In section 3, the polynomials arising from the operator $S$ are completely characterized. The result involve two families of polynomials for which the generating functions, difference equations, ladder operators, etc. are explicitly provided. The contraction limit is examined in all these instances and shown to correspond systematically to the characterization of the Meixner polynomials. In section 4, the same program is carried out for the operator $Q$; the contraction in this case leads to $d$-Charlier polynomials. We conclude with an outlook.
\section{The algebra $\mathfrak{su}(2)$, matrix elements and orthogonal polynomials}
In this section, we establish notations and definitions that shall be needed throughout the paper.
\subsection{$\mathfrak{su}(2)$ essentials}
\subsubsection{The $\mathfrak{su}(2)$ algebra and its irreducible representations}
The Lie algebra $\mathfrak{su}(2)$ is generated by three operators $J_{0}$, $J_{+}$ and $J_{-}$ that obey the commutation relations
\begin{align}
\label{su2-commutation}
[J_{+},J_{-}]=2J_{0}, && [J_{0},J_{\pm}]=\pm J_{\pm}.
\end{align}
The irreducible unitary representations of $\mathfrak{su}(2)$ are of degree $N+1$, with $N\in\mathbb{N}$. In these representations, $J_{0}^{\dagger}=J_{0}$ and $J_{\pm}^{\dagger}=J_{\mp}$, where $\dagger$ refers to the hermitian conjugate. We shall denote the orthonormal basis vectors by $\ket{N,n}$, for $n=0,\ldots,N$. The action of the generators on those basis vectors is
\begin{align}
\label{action-1}
J_{+}\ket{N,n}=\sqrt{(n+1)(N-n)}\,\ket{N,n+1},&& J_{-}\,\ket{N,n}=\sqrt{n(N-n+1)}\,\ket{N,n-1},
\end{align}
\begin{equation}
\label{action-2}
J_{0}\ket{N,n}=(n-N/2)\ket{N,n}.
\end{equation}
The operators $J_{\pm}$ will often be referred to as ''ladder operators''. Note that the action of $J_{-}$ and $J_{+}$ on the end point vectors is $J_{-}\ket{N,0}=0$ and $J_{+}\ket{N,N}=0$.

 It is convenient to introduce the number operator $\mathcal{N}$, which is such that $\mathcal{N}\ket{N,n}=n\ket{N,n}$; it is easily seen that this operator is related to $J_{0}$ by the formula $\mathcal{N}=J_{0}+N/2$. The most general action of any powers of the ladder operators $J_{\pm}$ on the basis vectors is expressible in terms of the Pochhammer symbol. Indeed, one finds
\begin{align}
\label{action-3}
J_{+}^{k}\ket{N,n}&=\sqrt{\frac{(n+k)!(N-n)!}{n!(N-n-k)!}}\ket{N,n+k}=\sqrt{(-1)^{k}(n+1)_{k}(n-N)_{k}}\ket{N,n+k},\\
\label{action-4}
J_{-}^{k}\ket{N,n}&=\sqrt{\frac{n!(N-n+k)!}{(n-k)!(N-n)!}}\ket{N,n-k}=\sqrt{(-1)^{k}(-n)_{k}(N-n+1)_{k}}\ket{N,n-k}.
\end{align}
These formulas are obtained  by applying \eqref{action-1} on the basis vector and by noting that $\frac{(n+k)!}{n!}=(n+1)_{k}$ and that $\frac{n!}{(n-k)!}=(-1)^k(-n)_{k}$ \cite{Koornwinder-2008}.
\subsubsection{Coherent states}
Let us introduce the $\mathfrak{su}(2)$-coherent states $\ket{N,\eta}$ defined as follows \cite{Perelomov-1986}
\begin{equation}
\label{coherent-states}
\ket{N,\eta}:=\sqrt{\frac{1}{(1+|\eta|^2)^N}}\;\sum_{n=0}^{N}\binom{N}{n}^{1/2}\eta^{n}\ket{N,n},
\end{equation}
where $\eta$ is a complex number. The action of the ladder operators on a specific coherent state $\ket{N,\eta}$ can be computed directly to find:
\begin{align}
\label{action-coherent}
J_{+}\ket{N,\eta}=\eta^{-1}\mathcal{N}\ket{N,\eta}, && J_{-}\ket{N,\eta}=\eta\,(N-\mathcal{N})\ket{N,\eta}.
\end{align}
These relations can be generalized to arbitrary powers of the ladder operators; one writes $\mathcal{N}$ in terms of  $J_{0}$ and uses commutation relations \eqref{su2-commutation} to obtain:
\begin{align}
\label{action-coherent-2}
J_{+}^{k}\ket{N,\eta}&=(-1)^{k}\eta^{-k}(-\mathcal{N})_{k}\ket{N,\eta}, \nonumber\\
J_{-}^{k}\ket{N,\eta}&=(-1)^{k}\eta^{k}(\mathcal{N}-N)_{k}\ket{N,\eta}.
\end{align}
Note that the formulas involving the Pochhammer symbols are to be treated formally.
\subsubsection{Contraction of $\mathfrak{su}(2)$ to $\mathfrak{h}_1$}
The contraction of $\mathfrak{su}(2)$ is the limiting procedure by which $\mathfrak{su}(2)$ reduces to the Heisenberg-Weyl algebra. The Heisenberg algebra is generated by the creation-annihilation operators $a^{\dagger}$, $a$ and the identity operator. This algebra is denoted $\mathfrak{h}_1$ and defined by the commutation relations
\begin{align}
\label{Heisenberg-commutation}
[a,a^{\dagger}]=\mathfrak{id},\;\text{and}\;\; [a,\mathfrak{id}]=[a^{\dagger},\mathfrak{id}]=0,
\end{align}
where $\mathfrak{id}$ stands for the identity operator. The contraction corresponds to taking the limit as $N\rightarrow\infty$; in order for this limit to be well defined, the ladder operators must be rescaled by a factor $\sqrt{N}$. The contracted ladder operators are denoted by
\begin{equation*}
J_{\pm}^{(\infty)}=\lim_{N\rightarrow\infty}\frac{J_{\pm}}{\sqrt{N}}.
\end{equation*}
In this limit, the action of the ladder operators on the basis vectors $\ket{N,n}$ tends to the action of the creation-annihilation operators in the irreducible representation of the $\mathfrak{h}_1$ algebra, which is infinite-dimensional. Indeed, the following formulas are easily derived:
\begin{align*}
\lim_{N\rightarrow\infty}\frac{J_{+}}{\sqrt{N}}\,\ket{N,n}&=\sqrt{n+1}\,\ket{n+1},\\
\lim_{N\rightarrow\infty}\frac{J_{-}}{\sqrt{N}}\,\ket{N,n}&=\sqrt{n}\,\ket{n-1}.
\end{align*}
This limit shall be used to establish the correspondence with studies associated to the Heisenberg-Weyl algebra;  see \cite{Vinet-2009}.
\subsection{Operators and their matrix elements}
The operators $S$, $Q$ and their matrix elements, which will be the central objects of study, are now defined.
\subsubsection{$S$, $Q$, their matrix elements and biorthogonality}
Let $a$ and $b$ be complex parameters; we define
\begin{align*}
S&:=e^{a J_{+}^2}e^{b J_{-}^2},\\
Q&:=e^{aJ_{+}}e^{bJ_{-}^{M}},
\end{align*}
with $M\in\mathbb{N}$. These operators can be represented by $(N+1)\times(N+1)$ matrices; their matrix elements are defined as
\begin{align*}
\psi_{k,n}&:=\Braket{k,N}{S}{N,n},\\
\varphi_{k,n}&:=\Braket{k,N}{Q}{N,n}.
\end{align*}
The two operators $S$, $Q$ are obviously invertible, with inverses given by $S^{-1}=e^{-bJ_{-}^2}e^{-aJ_{+}^2}$ and $Q^{-1}=e^{-bJ_{-}^{M}}e^{-aJ_{+}}$. We denote by $\chi_{n,k}=\Braket{n,N}{S^{-1}}{N,k}$ the matrix elements of $S^{-1}$; this leads to the following biorthogonality relation:
\begin{equation*}
\sum_{k=0}^{N}\chi_{m,k}\psi_{k,n}=\sum_{k=0}^{N}\Braket{m,N}{S^{-1}}{N,k}\Braket{k,N}{S}{N,n}=\Braket{m,N}{S^{-1}S}{N,n}=\delta_{nm},
\end{equation*}
where we have used the identity $\sum_{k=0}^{N}\ket{N,k}\bra{k,N}=\mathfrak{id}$, which follows directly from the orthonormality of the basis $\{\ket{N,k}\}_{k=0}^{N}$. A similar relation can be written for the matrix elements of $Q^{-1}$. If one defines $\varsigma_{n,k}=\Braket{n,N}{Q^{-1}}{N,k}$, then
\begin{equation*}
\sum_{k=0}^{N}\varsigma_{m,k}\varphi_{k,n}=\delta_{nm}.
\end{equation*}
\subsubsection{Recurrence relations and polynomial solutions}
The algebraic nature of the operators $S$ and $Q$ leads to recurrence relations satisfied by the matrix elements $\psi_{k,n}$ and $\varphi_{k,n}$. Let us start by showing how the recurrence relation for the $\psi_{k,n}$'s arises. The first step is to observe that
\begin{equation}
\label{recu-1}
(k-N/2)\psi_{k,n}=\Braket{k,N}{J_{0}S}{N,n}=\Braket{k,N}{SS^{-1}J_0S}{N,n}.
\end{equation}
The quantity $S^{-1}J_{0}S$ is computed using the Baker-Campbell-Hausdorff formula:
\begin{equation*}
e^{X}Ye^{-X}=Y+[X,Y]+\frac{1}{2!}[X,[X,Y]]+\frac{1}{3!}[X,[X,[X,Y]]]+\ldots.
\end{equation*}
Applying this formula and the relations of appendix A, one readily finds
\begin{equation*}
S^{-1}J_{0}S=J_{0}-2bJ_{-}^{2}+2a[J_{+}+2b(1+2J_{0})J_{-}-4b^2J_{-}^3]^2.
\end{equation*}
Substituting this result into \eqref{recu-1} yields for $\psi_{k,n}$ the  recurrence relation
\begin{align*}
2a\sqrt{(n+1)_{2}(n-N)_2}\,\psi_{k,n+2}&=(k-n)\,\psi_{k,n}+2b\sqrt{(-n)_2(N-n+1)_2}\,\psi_{k,n-2} \nonumber \\
&+2ab\sum_{t=0}^{3}(-b)^{t}\xi_{t}(n,N)\sqrt{(-n)_{2t}(N-n+1)_{2t}}\,\psi_{k,n-2t},
\end{align*}
where the coefficients $\xi_{t}(n,N)$ are given by:
\begin{align}
\label{coefficients-recu-1}
\xi_0(n,N)&:=2[2n-N][2n^2-2nN-N+1],\nonumber\\
\xi_{1}(n,N)&:=4[6n^2-6n(N+2)+N^2+5N+9],\nonumber\\
\xi_{2}(n,N)&:=16(2n-N-4),\\
\xi_{3}(n,N)&:=16\nonumber.
\end{align}
The recurrence relation for the matrix elements $\psi_{k,n}$ is not of the form \eqref{recu-ortho}. However, the jumps on the index $n$ are all even and it is clear from \eqref{action-2} that any $\psi_{k,n}$ with different index parity will be zero. Thus, setting $n=2j+q$, $k=2\ell+q$ with $q=0,1$ and introducing the monic polynomial $\psi_{k,n}=\psi_{k,q}\,\left(a^{-j}\sqrt{\frac{(N-n)!}{(N-q)!n!}}\,\right)\,\widehat{A}_{j}^{(q)}(\ell)$, the recurrence relation becomes 
\begin{align}
\label{recu-A}
\widehat{A}_{j+1}^{(q)}(\ell)&=(\ell-j)\widehat{A}_{j}^{(q)}(\ell)+c[(-n)_{2}(N-n+1)_2]\widehat{A}_{j-1}^{(q)}(\ell) \nonumber \\
&+c\sum_{t=0}^{3}(-c)^{t}\xi_{t}(n,N)\,[(-n)_{2t}(N-n+1)_{2t}]\;\widehat{A}_{j-t}^{(q)}(\ell),
\end{align}
with $c=ab$. This relation is precisely of the form \eqref{recu-ortho}; we thus conclude that the matrix elements $\psi_{k,n}$ are given in terms of two families of $d$-orthogonal polynomials with $d=3$ corresponding to $q=0$ and $q=1$ . These two families are fully characterized in the next section.

The recurrence relation for the matrix elements $\varphi_{k,n}$ can be obtained in the same fashion. It is first noted that
\begin{equation*}
(k-N/2)\,\varphi_{k,n}=\Braket{k,N}{J_{0}Q}{N,n}=\Braket{k,N}{QQ^{-1}J_0Q}{N,n}.
\end{equation*}
The computation of $Q^{-1}J_{0}Q$ proceeds along the same lines. The result is
\begin{equation}
\label{result-conjugaison}
Q^{-1}J_{0}Q=J_{0}+aJ_{+}-MbJ_{-}^M+abM\big(M-1+2J_{0}\big)J_{-}^{M-1}-ab^2M^2J_{-}^{2M-1}.
\end{equation}
Once again, introducing the monic polynomial $\varphi_{k,n}=\varphi_{k,0}\left(a^{-n}\sqrt{\frac{(N-n)!}{N!n!}}\right)\;\widehat{B}_{n}(k)$, this becomes
\begin{align}
\label{recu-B}
\widehat{B}_{n+1}(k)&=(k-n)\widehat{B}_{n}(k)+f\,\zeta_{M}(n,N)\,\widehat{B}_{n-M}(k) \nonumber\\
&+f\,\zeta_{M-1}(n,N)\,\widehat{B}_{n+1-M}
+f^2\,\zeta_{2M-1}(n,N)\,\widehat{B}_{n+1-2M},
\end{align}
with $f=a^{M}b$ and where the coefficients $\zeta(n,N)$ are: 
\begin{align*}
\zeta_M(n,N)&:=(-1)^{M}M(-n)_{M}(N-n+1)_{M},\\
\zeta_{M-1}(n,N)&:=(-1)^{M}M(-n)_{M-1}(N-n+1)_{M-1}(2n-M-N+1),\\
\zeta_{2M-1}(n,N)&:=(-1)^{2M-1}M^2(-n)_{2M-1}(N-n+1)_{2M-1}.
\end{align*}
Therefore, the polynomials $\widehat{B}_{n}(k)$ are $d$-orthogonal with $d=2M-1$.
\subsubsection{Contractions and the $\mathfrak{h}_{1}$ algebra}
It is relevant at this point to recall related results obtained in connection with the Heisenberg algebra $\mathfrak{h}_1$. These are to be compared with the contractions of the polynomials $\widehat{A}_{j}^{(q)}(\ell)$ and $\widehat{B}_{n}(k)$ obtained in the next sections.

In \cite{Vinet-2011}, two of us investigated the matrix elements
\begin{equation*}
\psi_{k,n}^{(\infty)}=\Braket{k}{e^{b( a^{\dagger})^2}e^{c\,a^2}}{n},
\end{equation*}
where $a$ and $a^{\dagger}$ are the generators of $\mathfrak{h}_1$ and the vectors $\ket{n}$ are the basis vectors of its irreducible representation. It was shown that these matrix elements are given in terms of two series of Meixner polynomials with different parameters. Indeed, with $n=2j+q$ and $k=2\ell+q$, one finds\footnote{See appendix B for definition and properties of Meixner polynomials $M_{n}(x;\beta,z)$.}
\begin{equation*}
\psi_{k,n}^{(\infty)}\propto M_{j}(\ell;1/2+q, z).
\end{equation*}
It is clear that this operator corresponds to the contraction of the operator $S$ previously defined. Thus, the polynomials $\widehat{A}^{(q)}_{j}(k)$ are expected to tend to the Meixner polynomials in the limit as $N\rightarrow \infty$ and can be interpreted as a $d$-orthogonal finite ``deformation`` of Meixner polynomials. 

In \cite{Vinet-2009}, we investigated the matrix elements
\begin{equation*}
\varphi_{k,n}^{(\infty)}=\Braket{k}{e^{\beta a^{\dagger}}e^{\sigma a^{M}}}{n}.
\end{equation*}
These matrix elements were found to be $d$-Charlier polynomials with $d=M$. These matrix elements correspond to the contraction of the matrix elements of the operator $Q$ just defined. Consequently, it is expected that the contraction of the operator $Q$ will lead to these $d$-Charlier. Moreover, the matrix elements $\varphi_{k,n}$ are expected to  yield the standard Krawtchouk polynomials $K_{n}(x;p,N)$ when\footnote{See appendix B for definition and properties of Krawtchouk polynomials $K_{n}(x;p,N)$.} $M=1$; the cases $M\neq 1$ correspond therefore to some $d$-Krawtchouk polynomials.
\section{Characterization of the $\widehat{A}_{j}^{(q)}(\ell)$ family}
We shall now completely characterize the family of $d$-orthogonal polynomials arising from the matrix elements of the operator $S=e^{aJ_{+}^2}e^{bJ_{-}^2}$; these matrix elements have already been shown to satisfy the recurrence relation \eqref{recu-A}. The general properties are computed first and contractions are studied thereafter.
\subsection{Properties}
\subsubsection{Explicit matrix elements}
We first look for the explicit expression of the matrix elements $\psi_{k,n}=\Braket{k,N}{S}{N,n}$. This expression is obtained by setting $n=2j+q$ and $k=2\ell+q$, expanding the exponentials in series, using the actions \eqref{action-3} and \eqref{action-4} and recalling the identity $(a)_{2n}=2^{2n}\left(\frac{a}{2}\right)_{n}\left(\frac{a+1}{2}\right)_{n}$ as well as $(2n+q)!=2^{2n}\,q!\,n!\,(q+1/2)_{n}$. Extracting the factor 
$$
\psi_{k,q}=\frac{a^{\ell}}{\ell !}\sqrt{\frac{(N-q)!k!}{(N-k)!}},
$$
and pulling out the normalization factor ensuring that the polynomials are monic $(a)^{-j}\sqrt{\frac{(N-n)!}{(N-q)!n!}}$ yields
\begin{equation}
\psi_{k,n}=\left((a)^{-j}\sqrt{\frac{(N-n)!}{(N-q)!n!}}\,\right)\widehat{A}_{j}^{(q)}(\ell;c,N)\;\psi_{k,q},
\end{equation}
where we have set
\begin{align}
\widehat{A}_{j}^{(q)}(\ell;c,N):=\frac{c^{j}}{j!}\frac{(N-q)!n!}{(N-n)!}\;{}_2F_{3}\left[
\begin{matrix}
-j & -\ell & \\
q+1/2 & \frac{q-N}{2} & \frac{q-N+1}{2}
\end{matrix};\frac{1}{16c}
\right],
\end{align}
with $c=ab$. It is understood that if $n$ and $k$ have different parities, the matrix element $\psi_{k,n}$ is zero.

We thus have an explicit representation of the polynomials $\widehat{A}_{j}^{(q)}(k;c,N)$  in terms of generalized hypergeometric functions. These polynomials satisfy the recurrence relation \eqref{recu-1}. Moreover, the fact that $\widehat{A}_{j}^{(q)}(k;c,N)$ is expressed as a ${}_2F_{3}$ indicates that these polynomials belong to the classification proposed in \cite{Cheikh-2008}. Indeed, it is clear that the singleton $\{-\ell\}$ is $1$-separable; consequently, the polynomials are of the form \eqref{polynomial-form} with $s=1$. Thus, the polynomials $\widehat{A}_{j}^{(q)}(k;c,N)$ are examples of $d$-orthogonal polynomials corresponding to the case $1$ of the classification with $s=1$ and $q=d=3$.
\subsubsection{Explicit inverse matrix elements}
The matrix elements $\chi_{n,k}$ of the inverse operator $S^{-1}=e^{-bJ_{-}^2}e^{-aJ_{+}^2}$ can also be evaluated explicitly; they can be computed either by directly expanding the exponentials or by inspection. Indeed, one finds the matrix elements of the inverse to be
\begin{equation}
\label{inverse-A}
\chi_{n,k}=\psi_{N-k,N-n}^{\star},
\end{equation}
with $\star$ denoting the substitutions $a\rightarrow -a$ and $b\rightarrow -b$. If $N$ is an even number of the form $N=2p+2q$, the inverse is given by:
\begin{equation}
\label{explicit-inverse}
\chi_{n,k}=(-1)^{j-\ell}\frac{(a)^{j-\ell}}{(p-\ell)!}\sqrt{\frac{(N-k)!n!}{k!(N-n)!}}\widehat{A}_{p-j}^{(q)}(p-\ell;c,N),
\end{equation}
with $n=2j+q$ and $k=2\ell+q$. If $N$ is an odd integer, the cases $q=0$ and $q=1$ have to be treated separately; since no further difficulty arise and $\chi_{n,k}$ is expressed directly in terms of $\psi_{n,k}$, we omit the details.
\subsubsection{Biorthogonality relation}
As pointed out in section 2.2.1, the matrix elements of the inverse $S^{-1}$ provide the polynomials entering in the biorthogonality relations of the $\widehat{A}^{(q)}_{j}(\ell;c,N)$ family. In view of \eqref{inverse-A}, this biorthogonality relation reads
\begin{equation}
\label{biortho-A-1}
\sum_{k=0}^{N}\psi_{k,n}\psi_{N-k,N-m}^{\star}=\delta_{nm}.
\end{equation}
If $N=2p+2q$, the biorthogonality relation involves two $\widehat{A}_{j}^{(q)}(\ell;c,N)$ with the same $q$; in the case where $N=2p+1$, the relation will involve two polynomials with different values of $q$.

Firstly,  set $N=2p+2q$ and $m=2j'+q$; the biorthogonality relation \eqref{biortho-A-1} becomes
\begin{equation*}
\sum_{\ell=0}^{p}w_{\ell}(p)\widehat{A}^{(q)}_{j}(\ell;c,N)\widehat{A}_{p-j'}^{(q)}(p-\ell;c,N)=(-1)^{j}\delta_{jj'}
\end{equation*}
with the weight $w_{\ell}(p)=\frac{(-1)^{\ell}}{\ell!(p-\ell)!}$.

 Secondly, set $N=2p+1$ and $m=2j'+q$;  the biorthogonality relation is then of the form
\begin{equation*}
\sum_{\ell=0}^{p}w_{\ell}(p)\widehat{A}^{(0)}_{j}(\ell;c,N)\widehat{A}_{p-j'}^{(1)}(p-\ell;c,N)=(-1)^{j}\delta_{jj'}.
\end{equation*}
A similar relation can be obtained with the $q=0$ and $q=1$ polynomials in a different order. Thus, the two classes of polynomials corresponding to the values $q=0$ and $q=1$ of the $\widehat{A}^{(q)}_{j}(\ell;c,N)$ family are interlaced when the degree of the representation is odd and independent if the degree is even.
\subsubsection{Generating function}
We now derive the generating function for the $\widehat{A}_{j}^{(q)}(\ell;c,N)$ polynomials. This derivation is based on the action of the operator $S$ on coherent states $\ket{N,\eta}$. Consider the following function:
\begin{equation}
G(k;\eta):=\frac{1}{\psi_{k,q}}\sqrt{\frac{a^{-q}(N-q)!}{N!}}\sum_{n=0}^{N}\binom{N}{n}^{1/2}\psi_{k,n}\;\eta^{n}.
\end{equation}
Upon substitution of the expression for the matrix elements, one finds
\begin{equation*}
G(k;\eta)=\sum_{n=0}^{N}\widehat{A}_{j}^{(q)}(\ell;c,N)\frac{(\eta/\sqrt{a})^{n}}{n!},
\end{equation*}
setting $G(k;\eta)$ as a generating function for the polynomials $\widehat{A}_{j}^{(q)}(\ell;c,N)$. This generating function $G(k;\eta)$ is in fact the matrix element of $S$ with respect to the coherent state $\ket{N,\eta}$. Indeed, with definition \eqref{coherent-states}, it follows that
\begin{align*}
G(k;\eta)&=\frac{1}{\psi_{k,q}}\sqrt{\frac{a^{-q}(N-q)!}{N!}}\sum_{n=0}^{N}\Braket{k,N}{S}{N,n}\braket{n,N}{N,\eta}\nonumber\\
&=\frac{1}{\psi_{k,q}}\sqrt{\frac{a^{-q}(N-q)!}{N!}}\;\Braket{k,N}{S}{N,\eta},
\end{align*}
where the normalization factor was omitted. The matrix element $\Braket{k,N}{S}{N,\eta}$ can be evaluated by using the spectral decomposition. One first decomposes the matrix element as:
$$\Braket{k,N}{S}{N,\eta}=\sum_{\mu=0}^{N}\Braket{k,N}{e^{aJ_{+}^2}}{N,\mu}\Braket{\mu,N}{e^{bJ_{-}^2}}{N,\eta}.
$$
The action of the ladder operators $J_{\pm}$ on the coherent states can then be used to obtain:
\begin{equation}
G(\ell;q;\eta)=(e^{i\pi/2})^{N+q}(-\eta\sqrt{b})^N\sum_{m=0}^{\lfloor (N-q)/2\rfloor}\frac{(1/\sqrt{c})^{2m+q}}{(2m+q)!}(-\ell)_{m}\;H_{N-2m-q}\left(\frac{e^{i\pi/2}}{2\eta \sqrt{b}}\right),
\end{equation}
 where $H_{n}$ is the standard Hermite polynomial\footnote{See appendix B for definition of the Hermite polynomials.}.
\subsubsection{Difference equation}
In a fashion dual to the way the recurrence relation was obtained, the difference equation satisfied by the matrix elements $\psi_{k,n}$ can be derived with the help of the Baker-Campbell-Hausdorff formula and the formulas from appendix A. First, observe that
\begin{equation*}
(n-N/2)\psi_{k,n}=\Braket{k,N}{SJ_{0}}{N,n}=\Braket{k,N}{SJ_{0}S^{-1}S}{N,n}.
\end{equation*}
The operator $SJ_{0}S^{-1}$ must be evaluated here. Using the formulas given in the appendix, one obtains
\begin{equation*}
SJ_{0}S^{-1}=J_{0}-2aJ_{+}^2+2b[J_{-}+2aJ_{+}(1+2J_{0})-4a^2J_{+}^3]^2.
\end{equation*}
Substituting the result, one finds
\begin{align*}
(n-k)\psi_{k,n}&=2b\sqrt{(k+1)_2(k-N)_{2}}\;\psi_{k+2,n}-2a\sqrt{(-k)_2(N-k+1)_2}\;\psi_{k-2,n}	\nonumber\\
&-2ab\sum_{t=0}^{3}(-a)^{t}\xi_{t}(k,N)\sqrt{(-k)_{2t}(N-k+1)_{2t}}\;\psi_{k-2t,n},
\end{align*}
where the coefficients $\xi_{t}(k,N)$ are those found in \eqref{coefficients-recu-1}.
The difference equation for the matrix elements $\psi_{n,k}$ can straightforwardly be turned into a difference equation for the family of polynomials $\widehat{A}_{j}^{(q)}(k;c,N)$. Indeed, we have
\begin{align}
(j-\ell)\widehat{A}^{(q)}_{j}(\ell;c,N)&=c\;\Omega_{\ell}(q;N)\;\widehat{A}_{j}^{(q)}(\ell+1;c,N)-\ell\, \widehat{A}_{j}^{(q)}(\ell-1;c,N)\nonumber\\
&-c\sum_{t=0}^{3}(-\ell)_{t}\xi_{t}(2\ell+q,N)\;\widehat{A}_{j}^{(q)}(\ell-t;c,N),
\end{align}
with
 $$
\Omega_{\ell}(N)=\frac{1}{\ell+1}(2\ell+q+1)_{2}(2\ell+q-N)_{2}.
$$
This difference equation can be written as an eigenvalue problem; denoting $\nabla f(x)=f(x)-f(x-1)$, $\Delta f(x)=f(x+1)-f(x)$ and using the well-known identities
\begin{align*}
f(x+t)&=\sum_{w=0}^{t}\binom{t}{w}\Delta^{w}f(x),\\
 f(x-t)&=\sum_{w=0}^{t}(-1)^{w}\binom{t}{w}\nabla^{w}f(x),
\end{align*}
the recurrence relation can be written as
\begin{align*}
\mathcal{H}^{(q)}\widehat{A}_{j}^{(q)}(\ell;c,N)=j\widehat{A}_{j}^{(q)}(\ell;c,N),
\end{align*}
with
\begin{equation*}
\mathcal{H}^{(q)}=\ell\;\nabla+c\;\Omega_{\ell}(q;N)\sum_{w=0}^{1}\Delta^{w}-c\sum_{t=0}^{3}(-\ell)_{t}\xi_{t}(2\ell+q,N)\sum_{w=0}^{t}(-1)^{w}\binom{t}{w}\nabla^{w}.
\end{equation*}
\subsubsection{Shift operators}
The shift operators can also be derived from the Baker--Campbell--Hausdorff formula and the formulas from appendix A.
We have
\begin{equation*}
\sqrt{(-n)_2(N-n+1)_2}\;\psi_{k,n-2}=\Braket{k,N}{SJ_{-}^2}{N,n}=\Braket{k,N}{SJ_{-}^2S^{-1}S}{N,n}.
\end{equation*}
One then easily obtains
\begin{equation*}
SJ_{-}^2S^{-1}=[J_{-}+2aJ_{+}(1+2J_{0})-4a^2J_{+}^3]^2.
\end{equation*}
Upon substitution in the previous equation, one gets
\begin{align}
\sqrt{(-n)_2(N-n+1)_2}\;\psi_{k,n-2}&=\sqrt{(k+1)_{2}(k-N)_{2}}\,\psi_{k+2,n}\nonumber\\
&-a\sum_{t=0}^{3}(-a)^t\xi_{t}(k,N)\sqrt{(-k)_{2t}(N-k+1)_{2t}}\,\psi_{k-2t,n},
\end{align}
where the coefficients $\xi_{t}(k,N)$ are those given in \eqref{coefficients-recu-1}.
Using the identities involving the difference operators introduced above, we obtain
\begin{equation*}
F^{(q)}\widehat{A}_{j}^{(q)}(\ell;c,N)=(-2j-q)_{2}(N-2j-q+1)_{2}\;\widehat{A}_{j-1}^{(q)}(\ell;c,N),
\end{equation*}
with
\begin{equation*}
F^{(q)}=\Omega_{\ell}(q;N)\sum_{w=0}^{1}\Delta^{w}-\sum_{t=0}^{3}(-\ell)_{t}\xi_{t}(k,N)\sum_{w=0}^{t}(-1)^{w}\binom{t}{w}\nabla^{w}.
\end{equation*}
The computation of the backward shift operator for the $\widehat{A}_{j}^{(q)}(\ell;c,N)$ is much more involved; the operator contains terms up to order $18$. The explicit expression will not be given.
\subsubsection{$d$-Orthogonality functionals}
From the recurrence relation \eqref{recu-1} computed in section $2$, we concluded that the polynomials $\widehat{A}_{j}(\ell;c,N)$ are $d$-orthogonal with $d=3$. We now provide the functionals with respect to which the $\widehat{A}_{j}(\ell;c,N)$ polynomials are orthogonal.

The initial step is to compute the recurrence relation satisfied by the matrix elements $\chi_{n,k}$ of the inverse operator. First, we note that
\begin{equation*}
(k-N/2)\;\chi_{n,k}=\Braket{n,N}{S^{-1}J_{0}}{N,k}=\Braket{n,N}{S^{-1}J_{0}SS^{-1}}{N,k}.
\end{equation*}
Proceeding along similar lines as before, we get
\begin{equation*}
S^{-1}J_{0}S=J_{0}-2bJ_{-}^2+2a[J_{+}+2b(1+2J_{0})J_{-}-4b^2J_{-}^3]^2.
\end{equation*}
We then find the recurrence relation to be
\begin{align}
(k-n)\;\chi_{n,k}&=2a\sqrt{(-n)_2(N-n+1)_{2}}\;\chi_{n-2,k}-2b\sqrt{(n+1)_2(n-N)_{2}}\;\chi_{n+2,k}\nonumber\\
&-2ab\sum_{t=0}^{3}(-b)^{t}\sigma_{t}(n,N)\sqrt{(n+1)_{2t}(n-N)_{2t}}\;\chi_{n+2t,k},
\end{align}
where the coefficients are:
\begin{align*}
\sigma_0(n,N)&:=2(2n-N)(1+2n(n-N)-N),\\
\sigma_{1}(n,N)&:=4(6n^2-6n(N+2)+N^2-7N+9),\\
\sigma_{2}(n,N)&:=16(2n-N+4),\\
\sigma_{3}(n,N)&:=16.
\end{align*}
From this recurrence relation, the following proposition can be stated.
\begin{proposition}
We have
\begin{equation*}
\chi_{2j+q,2\ell+q}=\sum_{i=0}^{2}Y_{i}^{(j)}\;\Xi_{i}(\ell),
\end{equation*}
where $\Xi_{i}(\ell)=\chi_{2i+q,2\ell+q}$ and $Y_{i}$ is a polynomial in the variable $\ell$. The degree of the polynomial is determined as follows. Suppose $j=3\gamma+\delta$ with $\delta=0,1,2$; then, if $i\leq\delta$, the degree of the polynomial is equal to $\gamma$, otherwise it is equal to $\gamma-1$.
\end{proposition}
\begin{proof}
This statement is obtained directly from the recurrence relation. Setting $n=2j+q=0$ yields the element $\chi_{2(3)+q,k}$ in terms of the elements $\chi_{2(2)+q,k}$, $\chi_{2(1)+q,k}$, $\chi_{q,k}$ with degree $1$ in $\ell$ in the case of $\chi_{0+q,k}$ and $0$ for the others. The proposition then follows by induction.
\end{proof}
With this statement, we define the following linear functionals 
\begin{equation*}
\mathcal{L}_{i}^{(q)}[f(x)]=\sum_{x=0}^{\lfloor (N-q)/2\rfloor}\frac{a^{x}}{x!}\sqrt{\frac{(2x+q)!}{(N-2x-q)!}}\;\Xi_{i}(x)\,f(x),
\end{equation*}
for $i=0,1,2$. From the biorthogonality relation, it follows that
\begin{equation}
\mathcal{L}_{i}^{(q)}[\ell^{\gamma}\widehat{A}_{j}^{(q)}(\ell;c,N)]\begin{cases}=0,&\text{if}\, j\geqslant 3\gamma+i+1;\\ \neq 0, & \text{if}\,j=3\gamma+i.\end{cases}
\end{equation}
With these relations, the $d$-orthogonality of the family of polynomials $\widehat{A}_{j}^{(q)}(\ell;c,N)$ is manifest.

The algebraic setting at hand has thus allowed to completely characterize the family of polynomials $\widehat{A}_{j}^{(q)}(\ell;c,N)$. The Baker-Campbell-Hausdorff relation and the formulas from appendix A were used to derive the recurrence relation, the difference equations as well as the forward shift operator. Moreover, the action of the operators $J_{\pm}$ on the basis vectors $\ket{N,n}$ of the $(N+1)$-dimensional irreducible representation of $\mathfrak{su}(2)$ was used to compute explicitly the expression for the polynomials $\widehat{A}_{j}^{(q)}(\ell;c,N)$. As it turns out, these polynomials are $d$-orthogonal with $d=3$. They have a simple representation as ${}_2F_{3}$ generalized hypergeometric function and fall into the classification of Ben Cheikh \emph{et al.} due to the $1$-separability of the singleton set $\{\ell\}$. The inverse elements were found using the symmetry of the vectors $\ket{N,n}$ with respect to the action of the operators $J_{\pm}$; those inverse elements proved useful in obtaining the three orthogonality functionals for the $\widehat{A}_{j}^{(q)}(\ell;c,N)$'s. In the next subsection, we will look at the contractions of the $\widehat{A}_{j}^{(q)}(\ell;c,N)$ and their structural formulas.
\subsection{Contractions}
We now turn to the evaluation of the contraction of the  $\widehat{A}_{j}^{(q)}(\ell;c,N)$. We recall that the procedure of contraction, explained in section $2$, corresponds to taking the limit as $N\rightarrow \infty$ after renormalizing the operators $J_{\pm}$ by a factor of $\sqrt{N}$.
\subsubsection{Contraction of the recurrence relation}
Let us first apply the contraction to the recurrence relation. The renormalization of the $\mathfrak{su}(2)$ generators is implemented by substituting $c$ by $c/N^2$. We denote
\begin{equation*}
\lim_{N\rightarrow \infty}\widehat{A}_{j}^{(q)}(\ell;c,N)= \widehat{M}_{j}(\ell),
\end{equation*}
provided that this limit exists. After straightforward manipulations, one finds that after the renormalization, the terms in $\widehat{A}_{j-3}^{(q)}(\ell;c,N)$ and $\widehat{A}_{j-2}^{(q)}(\ell;c,N)$ are of order $\mathcal{O}(N^{-2})$ and $\mathcal{O}(N^{-1})$, respectively. Consequently, they tend to zero in the limit as $N\rightarrow \infty$ and the recurrence relation becomes
\begin{align*}
\ell\;\widehat{M}_{j}(\ell)&=\widehat{M}_{j+1}(\ell)+\Big\{j-c\,[2+4(2j+q)]\Big\}\widehat{M}_{j}(\ell) \nonumber \\
&-\Big\{c(2j+q)(2j+q-1)[1-4c]\Big\}\,\widehat{M}_{j-1}(\ell).
\end{align*}
The initial $5$-term recurrence relation contracts to a $3$-term recurrence relation. To identify to which orthogonal polynomial this relation corresponds, we set $-4c=\frac{d}{1-d}$. The recurrence becomes
\begin{align}
\label{recu-Meixner}
\ell\;\widehat{M}_{j}(\ell)&=\widehat{M}_{j+1}(\ell)+\frac{1}{1-d}\Big\{j+d(j+q+1/2)\Big\}\widehat{M}_{j}(\ell) \nonumber\\
&+\frac{d}{(1-d)^2}\Big\{(j+q/2)(j+q/2-1/2)\Big\}\widehat{M}_{j-1}(\ell).
\end{align}
We recognize in \eqref{recu-Meixner} the normalized recurrence relation of the Meixner polynomials $M_{j}(\ell;\beta,d)$ with $\beta=q+1/2$, for $q=0,1$. Thus, as expected, we have
\begin{equation*}
\lim_{N\rightarrow \infty}\widehat{A}_{j}^{(q)}(\ell;c,N)=M_{j}\left(\ell,q+1/2,\frac{c}{c-4}\right).
\end{equation*}
\subsubsection{Contraction of the matrix elements}
It is relevant to examine directly the contraction of the explicit formulas obtained for the matrix elements $\psi_{n,k}$. To take this limit, one must first expand $\psi_{n,k}$ by writing the generalized hypergeometric function as a truncated sum. Then, using the same renormalization as in the previous contraction, one straightforwardly obtains
\begin{equation}
\lim_{N\rightarrow \infty}\psi_{n,k}=\frac{a^{\ell}\,b^{j}}{\ell!\,j!}\sqrt{k!\,n!}\;{}_{2}F_{1}\left[\begin{aligned}&-j  -\ell \\ &(q+1/2)  \end{aligned};\frac{1}{4\,c}\right].
\end{equation}
Here, the explicit expression of the Meixner polynomials in terms of Gauss hypergeometric functions is recovered. The contraction at this level exhibits clearly the relationship between the parameters of the polynomial family $\widehat{A}_{j}^{(q)}(\ell;c,N)$ and the Meixner polynomials. We have $\beta=q+1/2$ and  $1-\frac{1}{d}=\frac{1}{4c}$. This is the result found in \cite{Vinet-2011}.\footnote{To recover the full result, one has to take $a\rightarrow -\overline{\omega}$, $b\rightarrow \omega$ and eliminate the diagonal term in expansion of \cite{Vinet-2011}.}
\subsubsection{Contraction of the generating function and coherent states}
The contraction limit of the generating function $G(k;\eta)$ can also be taken. However, it is easy to see that in the limit $N\rightarrow\infty$, the coherent states $\ket{N,\eta}$ as given by \eqref{coherent-states} are ill-defined. Indeed, under the contraction of the $\mathfrak{su}(2)$ algebra, the radius of the Bloch sphere, on which the coherent states are defined, must also be taken to infinity. Therefore, a well-defined contraction of the coherent states requires to take the limit $N\rightarrow\infty$ with the renormalization $\eta\rightarrow \eta/\sqrt{N}$ \cite{Perelomov-1986}. Note that with this renormalization, the coherent states become eigenstates of the operator $J_{-}/\sqrt{N}$ in the limit $N\rightarrow\infty$.

Performing the same renormalization, the contraction of the generating function yields that of the normalized Meixner polynomials
\begin{equation*}
\lim_{N\rightarrow\infty}G(k;\eta)=\sum_{n=0}^{\infty}\widehat{M}_{j}(\ell,q+1/2,c)\frac{(\eta/\sqrt{a})^{n}}{n!}=e^{b\eta^2 }\;\eta^{q}\;{}_1F_{1}\left[\begin{aligned}&-\ell \\ q&+1/2\end{aligned};-\frac{\eta^2}{4a}\right].
\end{equation*}

\subsubsection{Contraction of inverse matrix elements and biorthogonality}
In section $3.1$, the symmetry of the irreducible representation of $\mathfrak{su}(2)$ was used to obtain that $\chi_{n,k}=\psi_{N-k,N-n}^{\star}$. This result is due to the fact that in this representation, there exists a vector $\ket{N,0}$ which is annihilated by $J_{-}$ and a vector $\ket{N,N}$ which is annihilated by the operator $J_{+}$. In the contraction limit, the symmetry of the representation is not preserved. Indeed, the vectors of the $\mathfrak{h}_{1}$ irreducible representation are also labeled by a positive integer $n$, but this integer is unbounded from above; more precisely, there exits no basis vector such that $a^{\dagger}\ket{n}=0$, but the relation $a\ket{0}=0$ still holds. Consequently, all relations involving the inverse matrix elements $\chi_{n,k}$ must be recalculated with the contraction  applied directly to the operators. We stress that this loss of symmetry under the limit $N\rightarrow \infty$ is the \emph{reason} for the drastic change of behavior of the polynomials $\widehat{A}_{j}^{(q)}(\ell;c,N)$ and for the fact that, in particular, the $d$-orthogonality reduces to the standard orthogonality.

To find the expression for the inverse matrix elements in the contraction limit, one can simply calculate, directly from the operators, the matrix elements of $S^{-1}=e^{-ba^2}e^{-a(a^{\dagger})^2}$ written as $\chi_{n,k}=\Braket{n}{S^{-1}}{k}$ or take the contraction of the recurrence relation. Meixner polynomials with changes in the arguments are found.

Other limits involving objects such as the difference equation, the forward shift operators can be taken in the same way, yielding their counterparts for the Meixner polynomial.
\section{Characterization of the $\widehat{B}_{n}(k)$ family}
In this section, we fully characterize the $d$-orthogonal polynomials, with $d=2M-1$, in terms of which the matrix elements of the operator $Q=e^{aJ_{+}}e^{bJ_{-}^M}$ are expressed; these polynomials have already been shown to obey the recurrence relation \eqref{recu-B}. 
\subsection{Properties}
\subsubsection{Explicit expression for the matrix elements of $Q$}
The matrix elements $\varphi_{k,n}=\Braket{k,N}{Q}{N,n}$ can be computed explicitly by expanding the exponentials in series and using the actions \eqref{action-3} and \eqref{action-4}. To express the matrix elements in terms of generalized hypergeometric functions, one needs the identities
$$
(a)_{Mn}=M^{Mn}\prod_{\beta=0}^{M-1}\left(\frac{a+\beta}{M}\right)_{n}\;\;\text{and}\;\;
(Mn+q)!=M^{Mn}\,q!\prod_{\beta=0}^{M-1}\left(\frac{q+\beta+1}{M}\right)_{n}.
$$
Setting $n=Mj+q$ with $q\in\{0,\ldots,M-1\}$, one finds
\begin{equation*}
\varphi_{k,n}=\frac{a^{k-q}b^{j}}{j!(k-q)!q!}\sqrt{k!n!}\sqrt{\frac{(N-q)!(N-q)!}{(N-k)!(N-n)!}}\;{}_{1+M}F_{2M-1}\left[\begin{matrix}-j & \{\alpha_{m}\}\\ \{\beta_m\}& \{\gamma_{m}\}\end{matrix};\frac{-1}{(Ma)^Mb}\right],
\end{equation*}
with $\{\alpha_{m}\}=(q-k+m)/M$, $\{\beta_{m}\}=(q+m+1)/M$ with $q+m+1=M$ excluded from the sequence and $\{\gamma_{m}\}=(q-N+m)/M$, with $m$ running from $0$ to $M-1$. 

To obtain the exact expression for the $\widehat{B}_{n}(k)$ polynomials, one must pull out the ``ground state''
 $$
\varphi_{k,0}=a^{k}\binom{N}{k}^{1/2},
$$ 
and the normalization factor $a^{-n}\sqrt{\frac{(N-n)!}{N!n!}}$ from the expression of the matrix elements. The final expression reads, with $f=a^{M}b$:
\begin{equation*}
\varphi_{k,n}=\widehat{B}_{n}(k;f,N)\left[a^{-n}\sqrt{\frac{(N-n)!}{N!n!}}\right]\varphi_{k,0},
\end{equation*}
with
\begin{equation}
\widehat{B}_{n}(k;f,N)=\frac{(-1)^{q}(f)^{j}(-k)_{q}}{j!q!}\frac{(N-q)!n!}{(N-n)!}\;{}_{1+M}F_{2M-1}\left[\begin{matrix}-j & \{\alpha_{m}\}\\ \{\beta_m\}& \{\gamma_{m}\}\end{matrix};\frac{-1}{M^{M}f}\right].
\end{equation}
The recurrence relation \eqref{recu-B} indicates that these polynomials are $d$-orthogonal with $d=2M-1$, but it is clear that the set $\{\alpha_m\}$ is $s$-separable only for $M=1$, it is $1$-separable in this case. However, with this value of $M$, the polynomials are expressed as Gauss hypergeometric functions and are simply related to the Krawtchouk polynomials. For $M=1$, the matrix elements are
\begin{equation*}
\varphi_{k,n}=a^{k}b^{n}\binom{N}{n}^{1/2}\binom{N}{k}^{1/2}K_{n}(k;p,N),
\end{equation*}
with $p=-ab$. For any other value $M$, the polynomials are $d$-orthogonal extensions of the Krawtchouk ones, but they fall outside the classification of \cite{Cheikh-2008}.
\subsubsection{Matrix elements of $Q^{-1}$}
The matrix elements of the inverse operator $\varsigma_{n,k}=\Braket{n,N}{Q^{-1}}{N,k}$ with $Q^{-1}=e^{-bJ_{-}^{M}}e^{-aJ_{+}}$ can also be computed directly. Just as for the matrix elements $\psi_{n,k}$ of the $S$ operator, the matrix elements $\varsigma_{n,k}$ of the inverse operator $Q^{-1}$ possess a reflection symmetry; we indeed find
\begin{equation}
\varsigma_{n,k}=\varphi_{N-k,N-n}^{\star},
\end{equation}
where $\star$ denotes the replacements $a\rightarrow -a$ and $b\rightarrow -b$. In terms of the matrix elements $\varphi_{k,n}$, this biorthogonality relation reads
$$
\sum_{k=0}^{N}\varphi_{k,n}\varphi_{N-k,N-m}^{\star}=\delta_{nm}.
$$
 Because of the asymmetric form of the operator, the explicit expression for the matrix elements $\varsigma_{n,k}$ heavily depends on the behavior of $k$ and $N$ modulo $M$. An exact expression would thus comprise $M^2$ cases of residues.

For definiteness, we shall set $M=2$ whenever general expressions cannot be found in closed-form. Note that this case is relevant in the contraction limit because such powers of the creation-annihilation operators appear in the oscillator algebra $\mathfrak{sch}_1$ \cite{Vinet-2011}.
\subsubsection{Biorthogonality relations}
The biorthogonality of the $\widehat{B}_{n}(k;f,N)$ polynomials can be written for any value of $M$. Indeed, one has
\begin{equation}
\sum_{k=0}^{N}w_{k}\widehat{B}_{n}(k;f,N)\widehat{B}_{N-m}(N-k;f',N)=(-1)^{n}\delta_{nm},
\end{equation}
with $w_{k}=\frac{(-1)^{k}}{k!\,(N-k)!}$ and $f'=(-1)^{M+1}f$. Note that when $M=1$, this is not exactly the orthogonality relation of the Krawtchouk polynomials. To obtain the orthogonality of the Krawtchouk polynomials from this equation, one must remove the normalization factor and use Pfaff's transformation. This transformation is only available for ${}_2F_{1}$ hypergeometric functions, thus, $M=1$ is the only case for which the biorthogonality relation degenerates into the standard orthogonality.
\subsubsection{Generating function}
The generating function for the $\widehat{B}_{n}(k;f,N)$ polynomials is obtained as in the previous section; it can be derived explicitly for any value of $M$. Define
\begin{equation}
G(k;\eta)=\frac{1}{\varphi_{k,0}}\sum_{n=0}^{N}\binom{N}{n}^{1/2}\varphi_{k,n}\eta^{n}.
\end{equation}
Substituting the expression for the matrix elements, one gets
\begin{equation*}
G(k;\eta)=\sum_{n=0}^{N}\widehat{B}_{n}(k;f,N)\frac{(\eta/a)^{n}}{n!},
\end{equation*}
yielding a generating function for the polynomials  $\widehat{B}_{n}(k;f,N)$. Once again, this generating function is expressed as the overlap between the vector $\ket{N,k}$ and $Q\ket{N,\eta}$. Indeed, it follows from the definition of the coherent states \eqref{coherent-states} that
\begin{equation*}
G(k;\eta)=\frac{1}{\varphi_{k,0}}\Braket{k,N}{Q}{N,\eta}.
\end{equation*}
Using the action of the ladder operators \eqref{action-coherent} and \eqref{action-coherent-2}, one easily finds
\begin{equation}
G(k;\eta)=\sum_{\mu=0}^{N}\frac{(-\eta/a)^{\mu}}{\mu!}(-k)_{\mu}\; {}_MF_{0}\left[\{\delta_{m}\};(-1)^{M}(M\eta)^{M}b\right],
\end{equation}
with $\{\delta_{m}\}=\frac{\mu-N+m}{M}$ with $m=0,\ldots,M-1$. After a slight adjustment in the parameters and use of the identity ${}_1F_{0}(-a;b)=(1-b)^{a}$, the generating function for the Krawtchouk (see Appendix B) polynomials can be recovered for $M=1$.
\subsubsection{Difference equation}
The algebraic setting can be used to derive the difference equation satisfied by the matrix elements $\varphi_{n,k}$. Again, one writes
\begin{equation*}
(n-N/2)\,\varphi_{k,n}=\Braket{k,N}{QJ_{0}}{N,n}=\Braket{k,N}{QJ_{0}Q^{-1}Q}{N,n}.
\end{equation*}
From the Baker--Campbell--Hausdorff relation, it follows that
\begin{equation*}
QJ_{0}Q^{-1}=J_{0}-aJ_{+}+Mb[J_{-}+2aJ_{0}-a^2J_{+}]^{M}.
\end{equation*}
Setting $M=2$, the difference equation for the matrix elements is found to be
\begin{align}
(n-k)\,\varphi_{k,n}&=2b\sqrt{(k+1)_2(k-N)_2}\,\varphi_{k+2,n}+4ab\,\sqrt{(k+1)(N-k)} \,\zeta_{1}\,\varphi_{k+1,n}\nonumber \\ &+2\,a^2\,b\,\zeta_{0}\,\varphi_{k,n}-4a^3b\sqrt{k(N-k+1)}\;\zeta_{1}\;\varphi_{k-1,n}\nonumber\\
&-a\sqrt{k(N-k+1)}\,\varphi_{k-1,n}
+2a^4b\sqrt{(-k)_2(N-k+1)_2}\,\varphi_{k-2,n},
\end{align}
with the coefficients $\zeta_{i}$:
\begin{align*}
\zeta_{0}&=6k^2-6kN+N(N+1),\\
\zeta_{1}&=2k-N-1.
\end{align*}
For the $\widehat{B}_{n}(k)$ polynomials, this equation becomes
\begin{align}
n\widehat{B}_{n}(k;f,N)&=(2f)(k-N)_{2}\widehat{B}_{n}(k+2;f,N)+(4\,f\,\zeta_1)(N-k)\widehat{B}_{n}(k+1;f,N)\nonumber\\
&+(k+2\,f\,\zeta_0)\widehat{B}_{n}(k;f,N)-k(4f\,\zeta_1+1)\widehat{B}_{n}(k-1;f,N)\nonumber\\
&+(2f)(-k)_2\widehat{B}_{n}(k-2;f,N).
\end{align}
Using the same identities as before, this difference equation can be written as an eigenvalue equation.
\subsubsection{Orthogonality functionals}
The $d$-orthogonality functionals can be computed for the polynomials $\widehat{B}_{n}(k;c,N)$. First we have:
\begin{equation*}
(k-N/2)\,\varsigma_{n,k}=\Braket{n,N}{QJ_{0}}{N,k}=\Braket{n,N}{QJ_{0}Q^{-1}Q}{N,k}.
\end{equation*}
For $M=2$, the formula \eqref{result-conjugaison} yields
\begin{equation*}
Q^{-1}J_{0}Q=J_{0}-2bJ_{-}^2+aJ_{+}+2ab(1+2J_{0})J_{-}-4ab^2J_{-}^3.
\end{equation*}
Substitution in the first equation gives a recurrence relation for the matrix elements of the inverse operator $Q^{-1}$; this relation is 
\begin{align}
(k-n)\;\varsigma_{n,k}&=a\sqrt{n(N-n+1)}\;\varsigma_{n-1,k}+2ab(2n+1-N)\sqrt{(n+1)(N-n)}\;\varsigma_{n+1,k}\nonumber\\
& -2b\sqrt{(n+1)_2(n-N)_2}\;\varsigma_{n+2,k}-4ab^2\sqrt{-(n+1)_3(n-N)_{3}}\;\varsigma_{n+3,k}.
\end{align}
The form of this recurrence relation suggests the following proposition:
\begin{proposition}
We have
\begin{equation}
\varsigma_{n,k}=\sum_{i=0}^{2}Y_{i}^{(n)}(k)\Xi_{i}(k),
\end{equation}
with $\Xi_{i}(k)=\varsigma_{i,k}$ and $Y_{i}^{(n)}$ a polynomial in $k$. If $n=3\gamma+\delta$ with $\delta=0,1,2$; then the degree of the polynomial is $\gamma$ when $i\leqslant\delta$ and $\gamma-1$ otherwise.
\end{proposition}
\begin{proof}
The proof follows from the recurrence relation and by induction.
\end{proof}
This suggests the definition of the linear functional
\begin{equation}
\mathcal{M}_{i} f(x)=\sum_{x=0}^{N}a^{x}\binom{N}{x}^{1/2}\Xi_{i}(x)f(x).
\end{equation}
It thus follows that 
\begin{align}
\mathcal{M}_{i}(k^{\gamma}\widehat{B}_{n}(k;c,N))\begin{cases}=0 & \text{if}\,n\geqslant 3\gamma+i+1\\
                                                 \neq 0 & \text{if}\, n=3\gamma+i
                                                 \end{cases}
\end{align}
General orthogonality functionals could be defined in full generality for the operator $Q^{-1}$. Indeed, one could emulate the procedure to build $2M-1$ orthogonality functionals.

\subsection{Contractions}
The contractions of the polynomials  $\widehat{B}_{n}(k)$ are expected to yield the $d$-Charlier polynomials considered in \cite{Vinet-2009}. The loss of symmetry also occurs in the contractions of the polynomials $\widehat{B}_{n}(k)$; consequently, contractions of the orthogonality functionals and ladder operators cannot be taken directly. As those cases were treated in full generality in \cite{Vinet-2009}, we simply show how the matrix elements $\varphi_{k,n}$ contract to the particular case of $d$-Charlier polynomials that was considered.
\subsubsection{Contraction of $\varphi_{k,n}$}
The required renormalization is $a\rightarrow a/\sqrt{N}$ and $b\rightarrow b/\sqrt{N^M}$. Once again, the limit is taken by first expanding the generalized hypergeometric function as a truncated sum and performing the indicated substitution. Upon simple transformations, the result is found to be
\begin{equation}
\lim_{N\rightarrow \infty}\varphi_{k,n}=\frac{a^{k-q}b^{j}}{j!(k-q)!q!}\sqrt{k!n!}\;{}_{1+M}F_{M-1}\left[\begin{matrix}-j & \{\alpha_m\}\\ \{\beta_{m}\} & -\end{matrix};\frac{(-1)^{M+1}}{a^{M}b}\right],
\end{equation}
which gives precisely the matrix elements obtained in \cite{Vinet-2009}.
\section{Conclusion}
We studied the $d$-orthogonal polynomials related to the classical Lie algebra $\mathfrak{su}(2)$. We showed the matrix elements of the operators $S=e^{aJ_{+}^2}e^{bJ_{-}^2}$ and $Q=e^{aJ_{+}}e^{bJ_{-}^{M}}$ were given in terms of two families of polynomials $\widehat{A}_{j}^{(q)}(\ell;c,N)$ and $\widehat{B}_{n}(k;c,N)$. Using the algebraic setting, we characterized these polynomials; their explicit expressions in terms of hypergeometric functions allowed to identify those that belong to the classification given in \cite{Cheikh-2008}. 

We also studied the contraction limit in which $\mathfrak{su}(2)$ is sent to the Heisenberg algebra $\mathfrak{h}_1$. We showed that the $\widehat{A}_{j}^{(q)}(\ell;c,N)$ tend to standard Meixner polynomials when $N\rightarrow \infty$. The $\widehat{A}_{j}^{(q)}(\ell;c,N)$ can therefore be seen as discrete $d$-orthogonal versions of the Meixner polynomials. In addition, it was shown that the polynomials  $\widehat{B}_{n}(k;c,N)$ are some $d$-orthogonal Krawtchouk polynomials and that they converge to the $d$-Charlier polynomials in the contraction limit. 

In \cite{Vinet-2011}, the exponentials of linear and quadratic polynomials in the generators of the $\mathfrak{h}_{1}$ algebra were unified to yield matrix orthogonal polynomials; these considerations were motivated by the link with the quantum harmonic oscillator. Similarly, considerations regarding the discrete finite quantum oscillator would suggest to convolute the corresponding matrix elements. While providing physically relevant amplitudes for the quantum finite oscillator, this study could lead to $d$-orthogonal matrix polynomials, such as considered in \cite{Sorokin-1997}. We plan to report elsewhere on this.
\phantom{\cite{Koekoek-2010}}
\section*{Acknowledgments}
A.Z. wishes to thank the CRM for its hospitality. V.X.G. has benefited from a scholarship from Fonds qu\'eb\'ecois de la recherche sur la nature et les technologies (FQRNT). The research of L.V. is supported in part through a grant from the Natural Sciences and Engineering Research Council (NSERC) of Canada.
\appendix
\section*{Appendix A--Useful formulas for $\mathfrak{su}(2)$}
The relation
\begin{equation*}
J_{0}J_{\pm}^n=J_{\pm}^{n}(J_0\pm n),
\end{equation*}
holds and  can be proven straightforwardly by induction on $n$. Using this identity and the relations $(J_{\pm})^{\dagger}=J_{\mp}$ as well as $J_0^{\dagger}=J_{0}$ , it follows that for $Q(J_{\pm})$ denoting a polynomial in $J_{\pm}$, one has
$$
[Q(J_{\pm}),J_0]=\mp J_{\pm}Q'(J_{\pm}),
$$
where $Q'(x)$ denotes the derivative with respect to $x$. The preceding formula and the Baker-Campbell-Hausdorff relation lead to the identity
$$
e^{Q(J_{\pm})}J_0e^{-Q(J_{\pm})}=J_0\mp J_{\pm}Q'(J_{\pm}).
$$
In addition, we have the relations
\begin{align*}
[J_{+},J_{-}^{n}]&=2nJ_{0}J_{-}^{n-1}+n(n-1)J_{-}^{n-1},\\
[J_{-},J_{+}^{n}]&=-2nJ_{+}^{n-1}J_{0}-n(n-1)J_{+}^{n-1},
\end{align*}
which can also be proved by induction on $n$. With the help of the previous identities, one obtains
\begin{align*}
[J_{+},Q(J_{-})]&=2J_{0}Q'(J_{-})+J_{-}Q''(J_{-}),\\
[J_{-},Q(J_{+})]&=-2Q'(J_{+})J_{0}-J_{+}Q''(J_{+}),
\end{align*}
From these formulas it follows that
\begin{align*}
e^{Q(J_{-})}J_+e^{-Q(J_{-})}=J_+ -2J_{0}Q'(J_{-})-J_{-}[Q''(J_{-})+Q'(J_{-})^2],\\
e^{Q(J_{+})}J_-e^{-Q(J_{+})}=J_- +2Q'(J_{+})J_{0}+J_{+}[Q''(J_{+})-Q'(J_{+})^2].
\end{align*}
\section*{Appendix B--Meixner, Hermite and Krawtchouk polynomials}
\subsection*{Meixner polynomials}
The Meixner polynomials have the hypergeometric representation
$$
M_{n}(x;\beta,d)={}_2F_{1}\left[\begin{aligned}-n&,-x\\ &\beta\end{aligned};1-\frac{1}{d}\right].
$$
They satisfy the normalized recurrence relation
$$
x\widehat{B}_{n}(x)=\widehat{B}_{n+1}(x)+\frac{n+(n+\beta)d}{1-d}\widehat{B}_{n}(x)+\frac{n(n+\beta-1)d}{(1-d)^2}\widehat{B}_{n-1}(x),
$$
with
$$
M_{n}(x;\beta,d)=\frac{1}{(\beta)_n}\left(\frac{d-1}{d}\right)^n\widehat{B}_{n}(x).
$$
\subsection*{Hermite polynomials}
The Hermite polynomials have the hypergeometric representation
$$
H_{n}(x)=(2x)^{n}\;{}_2F_{0}\left[\begin{aligned}-\frac{n}{2}&,\frac{1-n}{2}\\ &-\end{aligned};-\frac{1}{x^2}\right].
$$
\subsection*{Krawtchouk polynomials}
The Krawtchouk polynomials have the hypergeometric representation
$$
K_{n}(x;p,N)={}_2F_{1}\left[\begin{aligned}-n&,-x\\ -&N\end{aligned};\frac{1}{p}\right].
$$
Their orthogonality relation is
$$
\sum_{x=0}^{N}\binom{N}{x}p^{x}(1-p)^{N-x}K_m(x;p,N)K_{n}(x;p,N)=\frac{(-1)^{n}n!}{(-N)_n}\left(\frac{1-p}{p}\right)^{n}\delta_{nm}.
$$
They have the generating function
$$
(1+t)^{N-x}\left(1-\frac{1-p}{p}t\right)^{x}=\sum_{n=0}^{N}\binom{N}{n}K_{n}(x;p,N)t^{n}.
$$
For further details, see \cite{Koekoek-2010}.

\end{document}